\documentclass{llncs}

\usepackage{listings}
\usepackage[utf8]{inputenc}
\usepackage{fancyvrb}
\usepackage{amsmath}
\usepackage{graphicx}
\usepackage{subcaption}
\usepackage{hyperref}
\usepackage{sectsty}
\usepackage{titlesec}
\usepackage{amssymb}
\usepackage[dvipsnames]{xcolor}
\usepackage{tikz}
\usepackage{wrapfig}
\usepackage{array}
\usepackage{dirtytalk}
\usepackage{ifthen}

\usetikzlibrary{positioning, arrows.meta}

\newboolean{anon}
\setboolean{anon}{false}

\DeclareUnicodeCharacter{25CB}{$\circ$}
\DeclareUnicodeCharacter{03A3}{$\Sigma$}
\DeclareUnicodeCharacter{03C0}{$\pi$}
\DeclareUnicodeCharacter{03BB}{$\lambda$}
\DeclareUnicodeCharacter{2A3E}{;}
\newcommand{\catname}[1]{{\normalfont\textbf{#1}}}
\newcommand{\Cont}{\catname{Cont}}
\newcommand{\Type}{\catname{Type}}
\newcommand{\kleisli}{\mathbin{{>}\!{=}\!{>}}}
\newcommand{\cont}[3]{(#1 : #2\rhd#3\ #1)}
\newcommand{\IO}{\texttt{IO}}
\newcommand{\String}{\texttt{String}}
\newcommand{\CLI}{\mathit{CLI}}

\newcommand{\All}{\mathit{All}}

\newcommand{\eseq}{\ \mathit{>\!>}\ }

\newcommand{\Table}{\texttt{Table}}
\newcommand{\File}{\texttt{File}}
\newcommand{\FS}{\mathit{FS}}
\newcommand{\toptop}{(\top\rhd\top)}
\usepackage{fontawesome5}

\spnewtheorem{prop}[definition]{Proposition}{\bfseries}{\itshape}
\spnewtheorem*{intu}{Intuition}{\bfseries}{\upshape}

\spnewtheorem*{remarkno}{Remark}{\mdseries}{\itshape}

\ifthenelse{\boolean{anon}}{
    \newcommand{\formalisationURL}{}
}{
    \newcommand{\formalisationURL}{https://andrevidela.com/aplas-code}
}
\newcommand{\formalised}{{\color{NavyBlue!75!White}{\raisebox{-0.5pt}{\scalebox{0.8}{\faCog}}}}}
\newcommand{\flinkurl}[1]{\href{#1}{\formalised}}
\newcommand{\flink}[1]{\flinkurl{\formalisationURL\##1}}
\newcommand{\strstr}{(\String\rhd \String)}
\title{Container Morphisms for Composable Interactive Systems}
\ifthenelse{\boolean{anon}}{
    \author{Anonymous}
}{
    \author{André Videla}
    \institute{University of Strathclyde \email{andre.videla@strath.ac.uk}}
}

\begin{document}

\maketitle

\begin{abstract}
This paper provides a mathematical framework for
client-server communication that results in a modular and type-safe architecture.
It is informed and motivated by the software engineering practice of developing server backends
with a database layer and a frontend, all of which communicate with a notion of
request/response.
I make
use of dependent types to ensure the request/response relation matches and show how this idea
fits in the broader context of containers and their morphisms. Using the category of containers
and their monoidal products, I define monads on containers that mimic their functional programming
counterparts, and using the Kleene star, I describe stateful protocols in the same system.
\end{abstract}

\section{Introduction}

There is a plethora of tools to write server backends (Ruby on Rails, Django, NodeJS), but
those libraries do not draw from existing mathematical theories, such as the pi-calculus or other
process calculi, and therefore, do not enjoy a common understanding
driving their design. As a result, there is no overarching story around combining them
together.\\
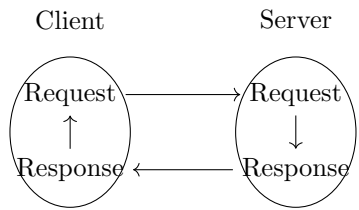
\begin{wrapfigure}{l}{0.5\textwidth}
    \vspace{-1em}
\centering
\begin{tikzpicture}
    \node (client) at (0,2) {Client};
    \node (server) at (3,2) {Server};
    \node (request-cli) at (0,1) {Request};
    \node (request-ser) at (3,1) {Request};
    \node (response-cli) at (0,0) {Response};
    \node (response-ser) at (3,0) {Response};

    \draw[->] (request-cli) -- (request-ser);
    \draw[->] (request-ser) -- (response-ser);
    \draw[->] (response-ser) -- (response-cli);
    \draw[->] (response-cli) -- (request-cli);

    \draw[-] (3,0.5) ellipse (0.8 and 1.0);
    \draw[-] (0,0.5) ellipse (0.8 and 1.0);

\end{tikzpicture}
\centering
    \caption{A binary client-server model, events read clockwise.}\label{fig:trad}
    \vspace{-1.7em}
\end{wrapfigure}
To work towards an implementable theory that explains what they share and how their common parts enable
them to be combined, I present a revised view of client-server communication that I
contrast with a more \say{traditional} view of client-server interactions, characterised by being binary (fig \ref{fig:trad}).
Indeed, in the binary view , clients send requests and expect matching
responses from those requests. Dually, servers await requests and provide responses to the client.
What this picture does not tell us is how those requests and responses get further processed by each end.
The lack of treatment for subsequent processing of request or response renders the binary
model fundamentally \emph{uncompositional}.
A compositional framework should not only explain client-server communication but also
handle hidden protocols within clients and servers, for example, a server calling another one
as part of its core functionality.

\def\lowest{1}
\def\lowmid{1.8}
\def\himid{3.8}
\def\highest{5}
\def\arrlen{0.2}

\newcommand{\msgarrowr}[3]
{\draw [->] (#1 -0.5+\arrlen,\himid) -- (#1 + 0.5 -\arrlen,\himid) node[midway, above] (#2) {#3};}
\newcommand{\msgarrowl}[3]
{\draw [<-] (#1 -0.5+\arrlen,\lowmid) -- (#1 + 0.5 -\arrlen,\lowmid) node[midway, above] (#2) {#3};}
\definecolor{pastelblue}{RGB}{173, 216, 230}
\definecolor{pastelyellow}{RGB}{253, 253, 150}
\definecolor{pastelgreen}{RGB}{176, 254, 230}
\definecolor{pastelpink}{RGB}{255, 182, 193}

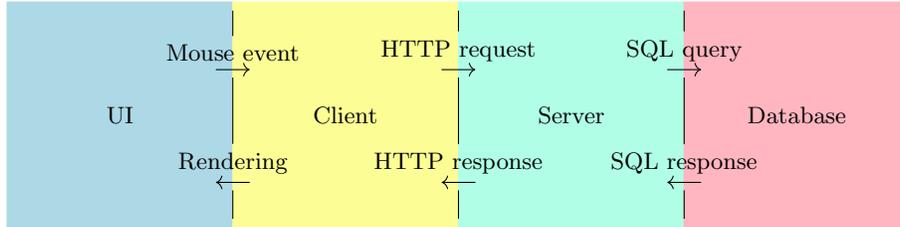
\begin{figure}[h]

\centering
    \begin{tikzpicture}[scale=0.75]

\filldraw[fill=pastelblue, draw=pastelblue] (-2, \lowest) rectangle (2,\highest);
\filldraw[fill=pastelyellow, draw=pastelyellow] (2, \lowest) rectangle (6,\highest);
\filldraw[fill=pastelgreen, draw=pastelgreen] (6, \lowest) rectangle (10,\highest);
\filldraw[fill=pastelpink, draw=pastelpink] (10, \lowest) rectangle (14,\highest);

\node (ui-cli-top) at (2, \highest) {};
\node (ui-cli-up-mid) at (2, 4) {};
\node (ui-cli-down-mid) at (2, \lowmid) {};
\node (ui-cli-bot) at (2, \lowest) {};

\node (cli-ser-top) at (6, \highest) {};
\node (cli-ser-up-mid) at (6, 4) {};
\node (cli-ser-down-mid) at (6, \lowmid) {};
\node (cli-ser-bot) at (6, \lowest) {};

\node (ser-db-top) at (10, \highest) {};
\node (ser-db-up-mid) at (10, 4) {};
\node (ser-db-down-mid) at (10, \lowmid) {};
\node (ser-db-bot) at (10, \lowest) {};

    \msgarrowr{2}{mid-up-left}{Mouse event}
    \msgarrowl{2}{mid-down-left}{Rendering}
    \msgarrowr{6}{mid-up-center}{HTTP request}
    \msgarrowl{6}{mid-down-center}{HTTP response}
    \msgarrowr{10}{mid-up-right}{SQL query}
    \msgarrowl{10}{mid-down-right}{SQL response}

\node at (0, 3) {UI};
\node at (4, 3) {Client};
\node at (8, 3) {Server};
\node at (12, 3) {Database};
\draw [-, shorten >=3pt] (ui-cli-bot) -- (mid-down-left);
\draw [-, shorten <=3pt] (mid-up-left) -- (mid-down-left);
\draw [-] (mid-up-left) -- (ui-cli-top);

\draw [-, shorten >=3pt] (cli-ser-bot) -- (mid-down-center);
\draw [-, shorten <=3pt] (mid-up-center) -- (mid-down-center);
\draw [-] (mid-up-center) -- (cli-ser-top);

\draw [-, shorten >=3pt] (ser-db-bot) -- (mid-down-right);
\draw [-, shorten <=3pt] (mid-up-right) -- (mid-down-right);
\draw [-] (mid-up-right) -- (ser-db-top);

\end{tikzpicture}
    \caption{An illustration of my revised compositional client-server view. Each color represents a different piece of software, each
    boundary between colors represents a protocol with messages and responses. The
    lifecycle of a single interaction can be read clockwise starting from the $UI$.}
    \label{new-client-server}
    \vspace{-1em}
\end{figure}

Figure~\ref{new-client-server} shows a simplified example of the interaction model I explore here, where events
emerge from a client, then are sent through multiple layers of software, and finally a response is produced.
In this revised view of client-server communication, the focus is on the interface between multiple
systems that communicate via different protocols. For a web application, the protocol is HTTP and the
interface is a list of endpoints, traditionally called the \say{Application Programming Interface}
or API. The server itself communicates with a database, the protocol is SQL, and the interface is
the language of valid SQL queries.

My claim is that containers and their morphisms~\cite{abbottCategoriesContainers2003}
accurately describe multi-tier applications like the
one pictured in Figure~\ref{new-client-server}.
Containers themselves describe the allowed request-response patterns between components.
Morphisms between containers describe how requests at one layer are translated to requests at the
next, and moreover how matching responses are translated back. We shall see how containers' ability
to express the dependency of allowed responses on requests means that we can not only construct rich
APIs but also describe in detail the implementation of those APIs in terms of other APIs.
We shall see how the structure of the category of containers plays an important role in guiding us
toward reusable abstractions~\ref{prop:cont-cat}. I have implemented this framework in Idris, which means
that anyone can compile and run those programs. To demonstrate the expressiveness of the framework I
construct an example To-do list application in Section~\ref{sec:demo-todo} and implement an API to
write files in a contemporary operating system in Section~\ref{sec:demo-fs}.
Additionally, category theory offers the means to talk about Kleisli composition, monads, and comonads, which are
primordial tools to compose effectful programs.

\subsection{Contributions}
My contributions are as follows:
\begin{enumerate}
    \item I provide a new interpretation of containers as APIs, or \say{process boundaries},
        and their morphisms as API-transformers.
    \item I give a notion of session via the Kleene star on containers~\ref{def:kleene}, and a notion of
        monadic computing~\ref{def:maybe-monad}.
    \item For the above structure, I give mechanised proofs for their properties. They are
        accessible via the \formalised\ icon\ifthenelse{\boolean{anon}}{(\textbf{link redacted for review, open index.html instead})}{}.
        For the proofs, I use Idris2~\cite{bradyIdrisGeneralpurposeDependently2013,bradyIdrisQuantitativeType2021}
        not only because it is an implementation of dependent type theory~\cite{martin-lofIntuitionisticTheoryTypes1975} but also because it provides bindings to common software libraries such as NodeJS, SQLite, and more.
    \item I employ this model by building a to-do app
        using the same definitions as the ones used for proofs. This approach results
        in a correct-by-construction running executable with the expected functionality. As an example
        of a stateful API, I show how to replicate the interaction model of a filesystem.

\end{enumerate}

\subsection{Related Work}

The work by Hancock and Setzer~\cite{hancockInteractiveProgramsDependent2000} showed us how to do
effectful programming using containers and dependent types. Their approach centers on
writing a program \emph{within} a container, representing trees of potential interaction paths
(what they call a \say{world}), rather than \emph{between}
containers. As we will see later~\ref{def:kleene}, programs whose specification
is given by a container do occur, and Hancock and Setzer's IO Trees could be used for that.
However, they won't be the focus of this work.
Hancock and Hyvernat~\cite{hancockProgrammingInterfacesBasic2009} worked on a category of
interfaces between processes but studied in the context of topology rather than application development.

In the work from Abbot, Altenkirch, Ghani, and
McBride~\cite{abbottDerivativesContainers2003,abbottCategoriesContainers2003,abbottContainersConstructingStrictly2005}, they
use containers to describe data structures, and operations on those data structures.
In it, containers are defined by
a set of shapes and a set of positions indexed by the shapes. I will use this terminology when
referencing containers outside of the request/response semantics I present here.

The mathematical tools used here are very similar to the ones used by the applied category
theory community. This work fits in the same family,
however, the semantics are established in the
software engineering practice of full-stack application development, rather than compositional
game theory, machine learning, or cybernetics.
Bolt et al. employ the category of lawless \emph{lenses},
for game theory~\cite{boltBayesianOpenGames2019} and Capucci et al. use it for
cybernetics~\cite{capucciFoundationsCategoricalCybernetics2022}.

\begin{figure}
    \vspace{-2em}
\begin{subfigure}[h]{0.4\linewidth}
\centering
\begin{align*}
    & \text{request} :\text{Type} \\
    & \text{response} : \text{request} \rightarrow \text{Type}
\end{align*}
    \caption{Our interpretation of the interface of an interactive system.}
\end{subfigure}
\hfill
\begin{subfigure}[h]{0.4\linewidth}
\centering
\begin{align*}
    & \text{response} : \text{Type} \\
    & \text{request} : \text{response} \rightarrow \text{Type}
\end{align*}
\caption{Spivak and Niu's interpretation of the interface for a dynamical system.}
\end{subfigure}%
    \vspace{-3em}
\end{figure}

Spivak and his team focus on the category $Poly$, the category of Polynomial
Functors and natural transformations between them, which is equivalent to our category of
containers. They also use it for cybernetics/dynamical
systems~\cite{niuPolynomialFunctorsMathematical2023} as well as
machine learning~\cite{fongBackpropFunctorCompositional2017}.

The biggest difference with their work on Poly lies in how we interpret containers as a type for requests/responses.
 What they consider responses, I think of as requests
and their notion of request is a notion of response here, the relationship is flipped.

The container morphisms presented here have a structure very similar to lenses as data
accessors for records.
So much so that they could be conceptualised as \emph{dependent lenses}, although they do not
provide access to a data structure. Instead, these morphisms move between process boundaries,
acting more like \emph{process accessors}, that grant access to a \emph{program}, rather than a data structure.
The previous work related
to the category of lenses~\cite{clarkeProfunctorOpticsCategorical2022}
shows that lenses take part in software development as functional data accessors for databases~\cite{bohannonRelationalLensesLanguage2006}
and data definitions~\cite{edwarda.kmettLensLibrary2024,oconnorFunctorLensApplicative2011,fosterCombinatorsBidirectionalTree2007}.
This versatility suggests that we could reuse the data accessor
aspect in this work, however, I will exclusively talk about processes and leave the integration
for future work.

The industry has produced a number of solutions to the problem described here.
Among those, we count the OpenAPI standard~\cite{OpenAPISpecificationV3}, GraphQL~\cite{GraphQL2021},
Typespec~\cite{Typespec2024}, Protobuf~\cite{ProtocolBuffers2024}, and more. These tools aim to simplify the task of writing client-server
software by providing a unique source of truth and relying on code generation to provide
implementations for both a client and a server. Similarly, libraries based on Object-Relational Mappings (ORMs~\footnote{Object–relational
mapping, Wikipedia:\\
\href{https://en.wikipedia.org/wiki/Object\%E2\%80\%93relational_mapping}{https://en.wikipedia.org/wiki/Object\%E2\%80\%93relational\_mapping}
}), like Django~\cite{Django2024} or GORM~\cite{GORM2024}, offer a way
to unify the database model and the data structure of the program, removing the API barrier between
the server code, and the database.
These tools achieve a small part of this work, to ease the communication between
binary processes.
However, they do not say anything about the software that fits in between the interfaces they
describe. What is more, because we do not have a formal definition of \say{interface} we cannot easily
talk about the compatibility between those systems.
A compositional framework ought to tell us
in which way two interfaces are incompatible and what is missing to make them compatible.

My goal here is to advance the state of the art by proposing a unifying theory that describes
composed systems interfacing with each other using different protocols. Additionally, this
theory needs to be usable in the context of a programming language to write those systems,
rather than purely describe them.

\section{Containers \& APIs}

The drive for this paper is the idea that containers are request/response pairs. These request/response pairs can be
conceptualised as the Application Programming Interface\footnote{ What is an API? by IBM:
\href{https://www.ibm.com/topics/api}{https://www.ibm.com/topics/api}} (API) of a given program. An API gives you what a
program can expect as inputs and what output to expect for a given input. For example, a C header file defines the API
of a C library, a list of endpoints defines the API of a microservice, etc. In what follows, I introduce the primary
tool for API description and manipulation: the category of containers. I will not provide an introduction to category
theory and will make use of functors, monads, comonads, and natural transformations\footnote{Resources for learning
category
theory:\href{https://github.com/prathyvsh/category-theory-resources}{https://github.com/prathyvsh/category-theory-resources}}.

\subsection{The category of containers}

Containers form a category,
and their morphisms allow to map from one container to another, in the interpretation of
containers as APIs, morphisms are \emph{API transformers}.

\begin{definition}[\flink{line23}]\label{def:container}
    A container consists of a type of requests, or queries $a : \Type$ and a type family of responses,
or results $a' : a \to \Type$.

\end{definition}

I write $(x : a \rhd a'\ x)$ to build the container with types $a : \Type$ and $a' : a \to \Type$.
When I use a built-in type without a binder like $(\String \rhd \String)$
I mean that the second argument is \emph{not} indexed by the first, it is equivalent to
$(x : \String \rhd \String)$. I write $\Type$ for the category of types and functions.

\begin{definition}[\flink{line30}]\label{def:cont-mor}
    Given two containers $a \rhd a'$, $b \rhd b'$,  a container morphism
    $f \lhd f' : a \triangleright a' \to b \triangleright b'$ consists of a function
    $f : a \to b$ mapping requests to requests, and a family of functions
    $f' : \forall (x : a). b' (f\ x) \to a'\ x$ mapping responses in the codomain to responses in the domain.
\end{definition}

\begin{intu} If the domain of the morphism is the "surface-level" API and the codomain is the "underlying" API,
then a morphism converts surface-level requests to underlying API requests. Additionally, it converts responses from the underlying
API back into responses that the surface-level API exposes.
\end{intu}

Because the second map of a morphism goes in the opposite direction as the first, I refer to the map on
requests as the \emph{forward part} of the morphism, and the map on responses as the
\emph{backward part} of the morphism, a terminology borrowed from lenses.
In the rest of the paper, I use the terms \emph{API transformer},
\emph{container morphism},  and \emph{lens},
interchangeably, and follow the naming convention of the related work.

We can develop the idea by studying an example. Let us imagine a web server with a client sending
HTTP requests, and the server returning HTTP
responses. We can define this API as the container:

\begin{align*}
    &\text{HTTP} : Container\\
    &\text{HTTP} = \cont{r}{\text{HTTPRequest}}{\text{HTTPResponse}}
\end{align*}

The implementation details of HTTP itself do not matter, we're only interested in providing the input/output relation here.

To implement such a server, we could provide a function $\cont{r}{\text{HTTPRequest}}{\text{HTTPResponse}}$ but this would be a very monolithic architecture.
Instead, we can see a server as a \emph{bidirectional program} built of multiple layers,
one of which could be a translation layer between HTTP request/response to database
query/response. If we assume there is a container to describe the API of a database
we could describe our server as a morphism:

\begin{minipage}{.5\linewidth}
\begin{equation*}
\begin{split}
    &\text{DB} : Container\\
    &\text{DB} = \cont{q}{\text{SQLQuery}}{\text{SQLResponse}}
\end{split}
\end{equation*}
\end{minipage}
\begin{minipage}{.5\linewidth}
\begin{equation*}
\begin{split}
    &\text{Server} : \text{HTTP} \Rightarrow \text{DB}\\
    &\text{Server} = \ldots
\end{split}
\end{equation*}
\end{minipage}
\\

Giving the implementation of a container morphism in one go still leads to a monolithic
architecture. Instead, we would like to build larger servers from smaller reusable
components. For this, we use the composition of container morphisms.

%
%
%

\begin{definition}[\flink{line37}]\label{def:cont-comp}
Given two morphisms $q_1 \lhd r_1 : a\Rightarrow b$ and $q_2 \lhd r_2 : b \Rightarrow c$, the composition
    of container morphisms $(q_1 \lhd r_1) ; (q_2 \lhd r_2)$ is given by $q_2 \circ q_1 \lhd \lambda x. r_1(x) \circ r_2
    (q_1(x)))$.
\end{definition}

\begin{intu}
    The main goal for composition is to translate requests from a surface-level API $(a \rhd a')$ down to
    an underlying API $(c \rhd c')$ \emph{through} an intermediate layer $(b \rhd b')$.
    This is achieved by function composition of the forward map.
    But the system also needs to \emph{translate back} responses from the underlying
    API into the surface-level responses. This is why we need to compose the backward map in reverse, starting
    from the underlying responses $c'$ we build responses back to the intermediate level $b'$ then back to $a'$.
\end{intu}
\begin{remarkno}
    There are other valid notions of composition for containers, this one matches what we expect from our
    API-transformer semantics.
\end{remarkno}

Using composition we can check that our structure is a category. The main benefit we get out of this
is the insurance that the systems we build remain composable regardless of their complexity.

\begin{prop}[{\flink{line45}}]\label{prop:cont-cat}
Containers form a category.
    Objects are containers~\ref{def:container}, morphisms are container morphisms~\ref{def:cont-mor},
    composition is given by morphism composition~\ref{def:cont-comp}, identity is given by the morphism $(id \lhd id)$.
\end{prop}
\begin{proof}
    Because a container morphism is essentially a pair of morphisms in $\Type$ and in
    $\Type^{op}$ the laws
    hold by inheriting the laws from the underlying category.
\end{proof}

When writing software, we often want to keep an internal representation of the data that we
wish to work with. As an example, we can design a server architecture that works with an internal
API $\cont{x}{InternalQuery}{InternalResponse}$ and build the software around it. Using
the composition of container morphisms we can attach morphisms to either end of this
internal API to obtain different results, as we see in figure~\ref{fig:reusable}.
\begin{figure}[h]
\begin{minipage}[t]{.5\linewidth}
    \begin{alignat*}{1}
        &\mathit{CLI} : \strstr \Rightarrow \mathit{Internal}\\
        &\mathit{CLI} = \ldots\\[1.2em]
        &\mathit{WEB} : \mathit{HTTP} \Rightarrow \mathit{Internal}\\
        &\mathit{WEB} = \ldots
    \end{alignat*}
\end{minipage}%
\begin{minipage}[t]{.5\linewidth}
    \begin{alignat*}{1}
        &\mathit{CLIApp} : \strstr \Rightarrow \mathit{DB}\\
        &\mathit{CLIApp} = \mathit{CLI} ⨾ \mathit{toQuery}\\[1.2em]
        &\mathit{WEBApp} : \mathit{HTTP}\Rightarrow \mathit{DB}\\
        &\mathit{WEBApp} = \mathit{WEB} ⨾ \mathit{toQuery}\\
    \end{alignat*}
\end{minipage}\vspace{-15pt}
    \caption{Two applications with a reusable component $toQuery : \mathit{Internal} \Rightarrow \mathit{DB}$ and two different translation layers converting from
    different surface-level APIs}\label{fig:reusable}
\end{figure}

The domain and codomain of morphisms \emph{CLI} or \emph{WEB} in figure~\ref{fig:reusable} are not
particularly realistic since they never fail. In reality, parsing an HTTP request, or
a string, can result in failure. So the API transformers need to model a
potential failing translation as well.

\subsection{The Maybe Monad on containers}\label{the-maybe-monad-on-containers}

Containers need to model the ability to fail or give partial results all while
using total functions to remain within the realm of provably correct code.
In the example of figure~\ref{fig:reusable} the API transformer
$\strstr \Rightarrow Internal$ has a forward map that converts
requests of type $\String \to InternalQuery$. But it
is unlikely that every possible string of characters produces a valid
query. Instead, we would like to represent the fact that
only some strings result in a valid message with the function
$\String \to \mathit{Maybe}\ InternalQuery$.
This suggests that the morphism $\mathit{CLI}$ should have the type:

\begin{equation*}
    \begin{split}
        & \CLI : \strstr \Rightarrow \cont{x}{\texttt{Maybe}\ InternalQuery}{InternalResponse}\\
        & \CLI = \ldots
    \end{split}
\end{equation*}

However, this does not typecheck because the $InternalResponse$
has type $InternalQuery \to \Type$ but we are supplying it a
value wrapped around a $\mathit{Maybe}$. To fix this, we use the $\All$ unary relation on $Maybe$.

\noindent
\begin{definition}[{\flink{line50}}]\label{def:maybe-all}\\
    \begin{minipage}[t]{.60\linewidth}
    Given a container $(q \rhd r)$, we define \\$\cont{m}{Maybe\ q}{All_r}$ as the MaybeAll
    map on containers. With All defined as follows:
\end{minipage}%
\hfill
\begin{minipage}[t]{.35\linewidth}
\vspace{-2.5em}
\begin{equation*}
\begin{aligned}
    & \forall (p : a \to \Type) \\
    & \All\ (\text{Just}\ x) \mapsto p\ x\\[0.1em]
    & \All\ \text{Nothing} \mapsto \top
\end{aligned}
\end{equation*}
\end{minipage}%
\end{definition}

Where $\top$ is the terminal object in $\Type$ with constructor $()$.
Using definition~\ref{def:maybe-all}, we can write the type of the $\mathit{CLI'}$ morphism as $\strstr \Rightarrow MaybeAll\ Internal$.

This starts to look like good old functional programming with a maybe
monad~\cite{wadlerMonadsFunctionalProgramming1993}. In fact, it is a monad in the category of containers:

\begin{prop}[\flink{line85}] MaybeAll is a functor in containers.
\end{prop}
    \begin{proof}
    Using definition~\ref{def:maybe-all} as the map on objects, we define
the map on morphisms:
    $(f\lhd f') \mapsto (map\ f \lhd mapAll\ f')$
Where $map$ is the standard map from the $Maybe$ type and $mapAll$ is
defined as:
    \begin{align*}
        &mapAll : (\forall x. p(f\ x) \to q\ x) \to  (x : Maybe\ a) \to All_p\ (map\ f\ x) \to All_q\ x\\
        &mapAll\ f'\ (Just\ x)\ v = f'\ v\\
        &mapAll\ f'\ Nothing\ () = ()
    \end{align*}
        The proof is given in appendix~\ref{sec:maybe-functor-proof}
\end{proof}
\begin{prop}[{\flink{line165}}]\label{def:maybe-monad}MaybeAll is a monad in $\Cont$.
\end{prop}
\begin{proof}
    We inherit the tools we need from the fact that $Maybe$ is a monad in $\Type$,
    with the $unit$ morphism given by $(Just \lhd id)$ and the $join$ by $(join \lhd id)$.
    By relying on the underlying $Maybe$ monad we ensure $unit$ behaves
    as an identity for $join$ and that $join$ is associative. The identity in the backward
    map is one thanks to pattern matching on the index.
\end{proof}

In the category of types and functions, we have that $Maybe\ x \approx 1 + x$
Naturally, one might ask whether this fact is also true in the category of containers.
To check this, we first
need to define the coproduct on containers.

\begin{definition}[{\flink{line176}}]\label{def:coproduct}\\
    \noindent
    \begin{minipage}[t]{.60\linewidth}
    Given two containers $(a \rhd a')$ and $(b\rhd b')$, the coproduct
    is given by $\cont{x}{a + b}{choice}$.
    \end{minipage}\hfill
    \begin{minipage}[t]{.35\linewidth}
        \vspace{-3.5em}
\begin{align*}
    &\forall a', b'. \\
    &choice\ (inl\ x) \mapsto a'\ x\\
    &choice\ (inr\ x) \mapsto b'\ x
\end{align*}
    \end{minipage}
\end{definition}
\begin{intu}The coproduct of two containers can be understood as building an
API that accepts requests from either \texttt{a} or \texttt{b} and returns
the corresponding answer. The $choice$ eliminator computes the return type
depending on what request was sent. If the request was of type $a$, then the
response should be of type $a'$, and similarly with $b$.
If the response was not dependent in this way we could be receiving responses
of type $b'$ after having sent a request of type $a$. We want to avoid that.
\end{intu}

Equipped with the coproduct, we can ask if there is an equivalent to the
isomorphism $Maybe\ x \approx 1 + x$ in $\Type$. Indeed, there is if
we define $1$ in containers by $\toptop$:

\begin{prop}[{\flink{line230}}]\label{prop:maybe-coprod}$MaybeAll\ x$ is Isomorphic to $1 + x$
\end{prop}
\begin{proof}
    We need two isomorphisms, one for the forward part, and one for the backward part, the forward
    isomorphism is given by the known functional programming result that
    $Maybe\ x \cong Either\ \top\ x$. For the backward part, we need to prove the isomorphism
    between $\forall (m : Maybe\ x). Any\ x'\ m$ and $\forall (e : \top + x). choice\ \bot\ x'\ e$.
    For the backward part, we need to pattern match on the index to find out how the type changes.
    In the $Just\ v$/$inr\ v$ case, the type is $x'\ v$ in both cases, when the index is
    $Nothing$/$inl\ ()$, the type is $\top$ in both cases.
\end{proof}

This indicates that our previous intuition
about types and functions does translate to containers and their morphisms. Including
features such as the diagonal map, which collapses two identical choices into one.

\begin{definition}[\flink{line236}]\label{def:dia}The diagonal map $a + a \Rightarrow a$ is given by the morphism
    $dia \lhd id$ where $dia$ is the diagonal on the coproduct in $\Type$.
\end{definition}

As usual, the identity in the backward part depends on the value of the index.
This representation makes the following map easy to write:
\begin{definition}[\flink{line241}] $MaybeU : MaybeAll\ 1 \Rightarrow 1$ removes redundant $MaybeAll$
\end{definition}

And is defined by composing the diagonal from definition~\ref{def:dia} with the coproduct isomorphism from
proposition~\ref{prop:maybe-coprod}.

Given a monad $M$, we can use Kleisli composition
$\kleisli : \forall a, b, c. (a \to M\ b) \to (b \to M\ c) \to (a \to M\ c)$
as a combinator to sequence multiple monadic programs. The following
example shows the use of Kleisli composition with the $MaybeAll$ monad
to sequence two potentially failing morphisms.

\begin{equation*}
    \begin{split}
        & parse : \strstr \Rightarrow MaybeAll\ \mathit{HTTP}\\
        & parse = \ldots\\[0.3em]
        & router : HTTP \Rightarrow MaybeAll\ \mathit{Internal}\\
        & router = \ldots\\[0.3em]
        & system : \strstr \Rightarrow MaybeAll\ \mathit{Internal}\\
        & system = parse \kleisli router
    \end{split}
\end{equation*}

\subsection{Stateful Protocols via the Sequential Product}\label{composition-as-sessions}

Some APIs expect the client to send multiple requests in a given order.
Such APIs are sometimes referred to as \say{stateful protocols}, and an instance
of successfully engaging in one is a \say{session}.

The sequential product, also known as the \emph{substitution product}, describes an API that expects
two queries, one after the other, and gives
responses for each of them. We could imagine sending a pair, but if we did that,
we could not choose what the second query should be depending on the first.
Capturing this dependency is at the core of the sequential product.
I use $\Sigma$-types instead of existential quantification to
match the representation of the program in code. Subscript $\pi1$ and
$\pi2$ are the projection functions from the $\Sigma$-type.

\begin{definition}[{\flink{line247}}]\label{def:ex-seq-prod}
    Given two containers $(a \rhd a')$ and $(b \rhd b')$, their
    sequential product $\eseq$
    is given by
        $(x : \Sigma (y : a) . a'\ y \to b \rhd
        \Sigma (z : a'\ x_{\pi1}) . b'\ (x_{\pi2}\ z))$ .
\end{definition}

\begin{intu} Sequencing two APIs should create one that will expect two
inputs in sequence and will provide both results at once. In the query
part, we see that we expect a query of type $a$,
but we also require a \emph{continuation} of type $a'\ y \to b$ that tells the system
the follow-up query given a response to $a$.
In the response, we send both results of each query, but because the second one
depends on the first, we return a $\Sigma$-type with the second response adjusted using the
continuation in the query.
\end{intu}

To effectively use sequential composition I introduce two maps to eliminate
its neutral element $1 = \toptop$, a container that will play a crucial role
in section~\ref{def:costate}.

\begin{definition}[\flink{line251}]
    $1$ is the neutral element for $\eseq$ given by the maps\\
\begin{minipage}{.45\linewidth}
\begin{equation*}
\begin{split}
    &UnitL : 1 \eseq x \Rightarrow x          \\
    &UnitL = (\pi 2\lhd \lambda x, y. ((), y))
\end{split}
\end{equation*}
\end{minipage}
\hfill
\begin{minipage}{.45\linewidth}
\begin{equation*}
\begin{split}
    &UnitR : x \eseq 1 \Rightarrow x\\
    &UnitR = (\pi 1\lhd \lambda x, y. (y, ()))
\end{split}
\end{equation*}
\end{minipage}
\end{definition}
\vspace{-0.4em}

For the proof that it is a monoidal product, refer to
Spivak~\cite{spivakReferenceCategoricalStructures2022}. I will only use this
as a combinator to build programs.

An example of a stateful protocol is a file system with a file handle that
needs to be accessed at first
and closed at the end of the interaction. If we think of \texttt{Open},
\texttt{Close}, and \texttt{Write} as standalone APIs to a filesystem,
then a successful interaction with the system is a sequence of
\texttt{Open} followed by \texttt{Write} and ended with \texttt{Close}:
$\texttt{Open} \eseq \texttt{Write} \eseq \texttt{Close}$.
We could even define a $Read$ container and write
$\texttt{Open} \eseq (\texttt{Write} + \texttt{Read}) \eseq \texttt{Close}$
if we wanted to indicate that the API allows both reads and writes
by using the coproduct~\ref{def:coproduct} on containers.

\subsection{Kleene Star for Repeated Requests}\label{sec:Kleene}

We can represent a complete interaction with the filesystem by the
composition of the three APIs $\texttt{Open} \eseq \texttt{Write} \eseq \texttt{Close}$
in a way that is impossible to mis-use the underlying protocol.
One odd
thing about this approach is that, with this API, we are
only allowed to do one thing: Write the file exactly once. If we want to
write to it multiple times, or if we want to close the file handle immediately after
opening it, we cannot, because it would be a breach of the API. What we really want
is to write something like this: $\texttt{Open} \eseq \texttt{Write}* \eseq \texttt{Close}$
where the postfix \texttt{*} is the Kleene star.
The Kleene star operator on containers is a repeated version of the sequential
product~\ref{def:ex-seq-prod} and mimics the way it works in regular expressions,
indicating 0 or more occurrences of a term strung together sequentially.

\begin{definition}[{\flink{line274}}]\label{def:kleene}
    Given a container $c = (a \rhd a')$ the Kleene star with type $\_* : Container \to Container$ is given using
    a data structure $StarShp : Container \to \Type$ and a map
    $StarPos : \forall c. StarShp_c \to \Type$ as $\cont{x}{StarShp_c}{StarPos_c}$
\begin{equation*}
\begin{split}
    & \texttt{data}\ StarShp : Container \to Type\ \texttt{where} \\
    & \quad Done : StarShp_{a\rhd a'}\\
    & \quad More : (x : a) \to (a'\ x \to StarShp_{a\rhd a'}) \to StarShp_{a\rhd a'}\\[0.5em]
    & StarPos : \forall (a\rhd a' : Container) . StarShp_{a\rhd a'} \to Type\\
    & StarPos\ Done = \top\\
    & StarPos\ (More\ req\ cont) = \Sigma (x : a'\ req) . StarPos_{a\rhd a'}(cont\ x)
\end{split}
\end{equation*}
\end{definition}

The Kleene star gives the choice of sending 0 API calls using
$Done$ or more requests by using $More$ to add to the list of requests to send.
We can build a smart constructor
$single : \forall (a \rhd a' : Container). a \to StarShp_{a \rhd a'}$ that will wrap one
layer of $More$ around a value of type $a$ and end the sequence with $Done$ effectively
sending a single request:
$single\ x \mapsto More\ x\ Done$.

\begin{prop}[{\flink{line343}}]
Kleene star is a functor.\end{prop}

\begin{proof}
    The map on objects is given by $\_* : Container \to Container$. The map on morphisms is given by:
\[
\begin{aligned}
& map_{*} : a \Rightarrow b \to a* \Rightarrow b*\\
& map_{*}\ m = \cont{x}{mapShp_m}{mapPos_m}
\end{aligned}
\]
Where, for all pairs of morphisms $a, b$, given a morphism $(f \lhd f') : a \Rightarrow b$
\[
\begin{aligned}
& mapShp : StarShp_a \to StarShp_b\\
& mapShp\ Done = Done\\
& mapShp\ (More\ x_1\ x_2) =
    More (f\ x_1) (mapShp \circ x_2 \circ f'_{x_1})
\end{aligned}
\]
\[
\begin{aligned}
& mapPos : \forall (x : StarShp_a) . StarPos_b (mapShp\ x) \to StarPos_a\\
& mapPos\ (x = Done)\ \top = \top\\
& mapPos\ (x = More\ x_1\ x_2)\ (y_1, y_2) =
    (f'_{x_1}\ y_1, mapPos_{x_2 (f'_{x_1}\ y_1)} y_2)
\end{aligned}
\]
To prove the functor preserves identities and composition we need to prove that both
$mapShp$ and $mapPos$ preserves identities and composition. That is, the four equations hold:

\begin{enumerate}
\item $mapShp_{id} = id$
\item $mapPos_{id} = id$
\item $mapShp_{f ; g} = mapShp_f ; mapShp_g$
\item $mapPos_{f ; g} = mapPos_f ; mapPos_g$
\end{enumerate}
    I prove each of these lemmas in appendix~\ref{apx:functor}.
\end{proof}

We have almost all the pieces to build large-scale systems
compositionally, but how do
we run them? That is the topic of the next section.

\subsection{Clients as State, Servers as Costate}

A container morphism is not a complete program, it is a tool to compose
different systems with matching APIs at their boundaries, but their definition
as pairs of maps does not explain how to turn them into executable programs.
In particular they are not \emph{closed}, their left and right boundary are
still open to further composition.

To close them, I borrow a technique that I learned from Jules Hedges in Open Games:
Contexts, State, and Costate. The terminology itself is adapted from quantum
computing~\cite{coeckePicturingQuantumProcesses2017}.

\begin{definition}[{\flink{line354}}]\label{def:state}
    A State-Lens is a morphism from $1$ into a container $(a \rhd a')$.
    It is isomorphic to a value of type $a$. $State\ c \triangleq 1 \Rightarrow c$\\
\begin{minipage}{.5\linewidth}
\begin{align*}
&state : a \to 1 \Rightarrow (a \rhd a')\\
&state\ x = (\lambda \_. x \lhd \lambda \_ \_ . \top)
\end{align*}
\end{minipage}
\begin{minipage}{.5\linewidth}
\begin{align*}
&value : 1 \Rightarrow (a \rhd a') \to a \\
&value\ (f \lhd \_) = f\ \top
\end{align*}
\end{minipage}
\end{definition}

\begin{definition}[{\flink{line368}}]\label{def:costate}A Costate-Lens is a morphism from
    $(a \rhd a')$ into $1$ and is isomorphic to a function $(x : a) \to a' x$. $Costate\ c \triangleq c \Rightarrow 1$
\begin{align*}
    &costate: ((x: a) \to a'\ x) \to (a \rhd a') \Rightarrow 1\\
    &costate\ f = (\lambda\_. \top, \lambda x\ \_ . f\ x)\\[0.3em]
    &exec : (a \rhd a') \Rightarrow 1 \to (x : a) \to a'\ x\\
    &exec\ (\_ \lhd f)\ x = f\ x\ \top
\end{align*}
\end{definition}
\begin{definition}[{\flink{line381}}] A Context for a lens
    $(a \rhd a') \Rightarrow (b \rhd b')$ is
    a pair $a \times ((x : b) \to b' x)$, or equivalently, a state for $(a \rhd a')$ and
    a costate for $(b \rhd b')$.
\end{definition}
\begin{intu}
    The original intuition for containers as APIs is that a container $(a \rhd a')$ gives us
    the type signature of a system that takes requests of $a$ and return responses $a'$.
    The definition of the Costate-Lens materialises this idea by turning it into a morphism.
    Similarly, for requests, if an API $(a \rhd a')$ accepts requests of type $a$, it
    should be that communicating with it requires producing values of type $a$.
\end{intu}

Once taken together, state and costate give us a way to \emph{run} morphisms in the same way
they are run in Open Games~\cite{boltBayesianOpenGames2019}.

\begin{equation*}\label{eq:runner}
    \begin{split}
        &run : (st : State (a \rhd a')) \to Costate\ (b \rhd b') \to (a \rhd a') \Rightarrow b \to a' (value\ st)\\
        &run\ st\ co\ m = exec (m ; co) (value\ st)
    \end{split}
\end{equation*}

\subsection{IO and Side-Effects}\label{state-and-side-effects}

We are now able to use our API transformers to extract programs using
a state, costate, and contexts.
But something is still missing. In everything I showed until now, I only ever described pure
functions,
but a tool like a database library will
not expose a pure function as its API. Rather, it will perform side effects, such as
\IO\ or exceptions. How to model effectful programs is the topic of this section.

Assuming a container for a database with API $(q : \texttt{DBQuery} \rhd \texttt{DBRes}\ q)$,
it is unlikely a library implementing it would be a pure function. Rather, it would be an
effectful one working in \IO\ with type $(q : \texttt{DBQuery}) \to \IO\ (\texttt{DBRes}\ q)$.

If we were to model this stateful function as a container it would be the new container
$\texttt{DBIO} = (q : \texttt{DBQuery} \rhd \IO\ (\texttt{DBRes}\ q))$.

As with the Maybe monad on containers, we can define a map on containers, lifting from monads on Type
to comonads on containers. Let us first see how to lift functors on Type to
functors on containers:

\begin{definition}[{\flink{line393}}]\label{def:lift}Given a functor $f : \Type \to \Type$, we define
\texttt{Lift}, the map on containers  from $(q \rhd r)$ to $(q \rhd f\circ r)$.
\end{definition}

\begin{prop}[{\flink{line416}}]
    Given a functor $f : \Type \to \Type$, $\texttt{Lift}\ f$ is an endofunctor in $\Cont$.
    The map on objects is given by $\texttt{Lift}\ f$, the map on morphisms is given by:
            $(g \lhd g') \mapsto g \lhd map\ f \circ g'$
    \label{f-f-functor}
\end{prop}
\begin{proof}
    The proof follows from the fact that $f$ is a functor, with the subtlety that the backward
    part of the morphism lives in $\Type^{op}$, and that duality also preserves the functor laws.
\end{proof}

This construction matches the one given by Abou-Saleh et al. on
\emph{monadic lenses}~\cite{abou-salehReflectionsMonadicLenses2016a}. Seeing it as a
map on containers sheds light on an interesting
fact about monadic programming in the category of lenses: Given a monad in $Set$, we obtain a
comonad in $Cont$.

\begin{prop}[{\flink{line464}}]
    Given a monad in $\Type$ $(T : \Type \to \Type, unit : \forall a . a \to T\ a,
    mult: \forall a. T\ (T\ a) \to T\ a)$ with the appropriate laws, we build the endofunctor
    $\texttt{Lift}\ T: \Cont \to \Cont$ like above~\ref{f-f-functor}.
    The counit is given by the morphism $(id \lhd unit)$ and the comultiplication
    by $(id \lhd mult)$.
\end{prop}

\begin{proof}
    Because only the backwards parts see any changes from $Lift\ T$, and it runs backward,
    the proof that $counit$ and $comult$ form a comonad follows from the fact that the dual
    of a monad is a comonad in the opposite category.
\end{proof}

Using the fact that $Lift\ f$  is a functor, we can map a morphism
$HTTP \Rightarrow DB$, and into an effectful
program $Lift\ \IO\ HTTP\ \Rightarrow Lift\ \IO\ DB$ to interact with other effectul APIs.
The counit gives us a way to interface
a stateful API with a pure one. For example, a test database that runs in memory could
have a pure signature, but the rest of the program will still expect to run with an \IO
\ effect, we can build the API transformer $Lift\ \IO\ DB \Rightarrow DB$ to plug
our program into a pure implementation of a database for testing purposes.

Finally, one handy tool is that $Lift\ f$ distributes around $+$, allowing us to
combine effectful computation with ones that provide a choice of API.
\begin{definition}[\flink{line475}]\label{def:distrib-plus}Given $f : \Type \to \Type$ an endofunctor,
    $distrib_+: Lift\ f\ (a + b) \Rightarrow Lift\ f\ a + Lift\ f\ b$ is given by
    the morphism $id \lhd id$.
\end{definition}

The backward part is a bit tricky because it is only an identity \emph{after}
we pattern match on its index. Because of the relationship between coproducts
and the $MaybeAll$ map, we define a similar distributivity map for it around
$Lift$:

\begin{definition}[\flink{line480}]\label{def:distrib-maybe}Given $f : \Type \to \Type$ an endofunctor,
    $distrib_{Maybe}: Lift\ f\ (MaybeAll\ a) \Rightarrow MaybeAll\ (Lift\ f\ a)$ is given by
    the morphism $id \lhd id$.
\end{definition}

Again the backward part is only an identity \emph{after} pattern matching.

Our last combinator allows to combine multiple costates sequentially using
monadic sequencing.

\begin{definition}[\flink{line486}]\label{def:seq-monad}For a given monad $m : \Type \to \Type$,
    and two containers $a, b : \Cont$,
    we can run two effectful costates
    $m1 : Costate\ (Lift\ m\ a)$ and $m2 : Costate\ (Lift\ m\ b)$ in sequence
    to obtain a combined effectful API $Costate (Lift\ m\ (a \eseq b))$.
\end{definition}

The implementation is straightforward in its mechanised version.

\section{Demo - Executing Interactive programs}\label{sec:demo}

It remains to demonstrate that this model is not only useful for abstractly talking about
properties of larger compositional systems, but also for their implementation. To this end,
I present two examples. First, a To-do app that communicates with a database, to show that the
system is successfully compositional across multiple domains. Second, I implement
the interface of a filesystem to show how stateful protocols can be described and run.

\subsection{A Basic To-do App}\label{sec:demo-todo}

A to-do app is a tool to manage lists of actionable items. Operations include creating a
new item, marking an item as done and retrieving the list of all items. For space
reasons, I will only show
small relevant parts of the original source code. For the full source code, you can find it here:
\ifthenelse{\boolean{anon}}{
    \textbf{redacted for review, open SQL.html instead}
}{
    \href{https://andrevidela.com/aplas-code/SQL.html}{https://andrevidela.com/aplas-code/SQL.html}.
}

The first choice to make is the model of interaction. A real app would use an HTTP interface
to expose the functionality of the application. For simplicity, I will use a command-line
interface via a Read-Eval-Print Loop (REPL).

\begin{equation*}
    \texttt{repl} : (\String \to \IO\ \String) \to \IO\ \top
\end{equation*}

This type suggests that a program interacting with the REPL has the API $\strstr$,
I call this container $\CLI$.
The app handles 3 messages : \texttt{Create}, \texttt{MarkDone}, \texttt{ListAll}.
As you can tell by their name, they allow the creation, update, and retrieval of todo items. Here
is a data definition for it:
\noindent
\[
    \texttt{data}\ TodoCommand = Create\ \String\ \mid\ MarkComplete\ \mathit{ID}\ \mid\ RetrieveAll
\]

For each message, we produce the type of responses.
Except for $RetrieveAll$, all messages do not return any value. Instead,
we will rely on effect lifting to perform side-effects for them.
\begin{equation*}
    \begin{split}
        &TodoResponse : TodoCommand \to Type\\
        &TodoResponse\ RetrieveAll = \Table\ TodoItem\\
        &TodoResponse\ \_ = \top
    \end{split}
\end{equation*}

Where $\Table$ is the type of SQL tables returned by SQL queries. These two types form the
internal API of the app: $App = \cont{cmd}{TodoCommand}{TodoResponse}$.

To implement the database, I use SQLite3 binders in Idris. The library has two methods of interactions:
Commands, and queries. They are implemented with the functions
$runCmd : \texttt{Cmd} \to \texttt{IO}\ \top$ and
$runQuery : (q : \texttt{Query}) \to \texttt{IO} (\Table\ q.type)$ respectively.
These functions suggest a costate with an appropriate container. For
space and legibility reasons, I write $\IO$ instead of $Lift\ \IO$:
\noindent
\begin{minipage}{.43\linewidth}
\begin{equation*}
\begin{split}
& SQLCmd : Container\\
& SQLCmd = (\texttt{Cmd} \rhd \top)\\[0.5em]
& execCmd : Costate (\IO\ SQLCmd)\\
& execCmd = costate\ runCmd
\end{split}
\end{equation*}
\end{minipage}
\hfill
\begin{minipage}{.55\linewidth}
\begin{equation*}
\begin{split}
& SQLQuery : Container\\
    & SQLQuery = (q : Query \rhd \Table\ q.type)\\[0.5em]
    & execQry : Costate\ (\IO\ SQLQuery)\\
& execQry = costate\ runQuery\\
\end{split}
\end{equation*}
\end{minipage}\\[0.3em]

Each costate is combined with a coproduct to give the choice of the caller to interact
with the database by supplying either a command or a query.
\begin{equation*}
\begin{split}
& execDB : Costate (\IO\ (SQLCmd + SQLQuery))\\
& execDB = distrib_{+} ; (execCmd + execQuery) ; dia
\end{split}
\end{equation*}
Where $distrib_{+}$ distributes effects across the coproduct~\ref{def:distrib-plus} and
$dia$ is the diagonal operator~\ref{def:dia}.
We have all the \emph{interface} pieces, now we need to implement the morphisms to translate
between one API to the next. This is done via two morphisms, one that I call $parser$ but
performs two operations, parse incoming requests but also print responses for the command
line to show. The second, called $toDB$, converts internal messages into database
commands and queries, and converts backs responses from the database into internal
messages.
\noindent
\begin{minipage}{.4\linewidth}
\begin{equation*}
    \begin{split}
        &parser : \CLI \Rightarrow MaybeAll\ App\\
        &parser = parse \lhd print
    \end{split}
\end{equation*}
\end{minipage}
\hfill
\begin{minipage}{.55\linewidth}
\begin{equation*}
    \begin{split}
        &toDB : App \Rightarrow (DBCmd + SQLQuery)\\
        &toDB = \ldots
    \end{split}
\end{equation*}
\end{minipage}\\[0.3em]

I elide the implementation of $parse$, $print$, and $\mathit{toDB}$ since their details are not relevant
to this example.

It is worth noting that the mapping $toDB$ is pure because, at this stage, we are only
dealing with valid messages. Accordingly, all responses from the database will also
be valid responses. Therefore, there is no need to concern ourselves with errors yet.
It is only once we compose the API transformer with an effectful database implementation
that we start worrying about how to handle potential errors and effects.

With all those pieces, we can build the
entire system by composing  the parser, the message conversion,
and the database execution.

\begin{equation*}
    \begin{split}
        &app : Costate (\IO\ CLI)\\
        &app = map_{IO}(parser) ; distrib_{Maybe}  ; map_{Maybe}(map_{IO}(toDB) ; execDB) ; MaybeU\\
    \end{split}
\end{equation*}

This makes use of composition~\ref{def:cont-comp}, distributivity~\ref{def:distrib-maybe},
the $MaybeAll$ functor on containers~\ref{def:maybe-all}, and the $Lift$ functor~\ref{def:lift}.
The extra piece at the end cleans up the bureaucracy brought on by $MaybeAll$.

The result of this is $app$, an implementation of the entire pipeline from parsing to database
queries, going via internal messages and a bespoke database interaction model.

\subsection{A filesystem Session}\label{sec:demo-fs}
\setlength{\abovedisplayskip}{0.4em} 
\setlength{\belowdisplayskip}{0.4em} 
\setlength{\abovedisplayshortskip}{0.4em} 
\setlength{\belowdisplayshortskip}{0.4em} 
In section~\ref{composition-as-sessions}, I claimed that the sequential product
describes sessions and that their type constrains the sequence of
requests that need to be sent to constitute a valid
session. The example I gave was the API for a filesystem with three
different messages for opening a file handle, one for writing to a
file handle, and one for closing a file handle. I will now elaborate on this
example here with some code snippets to show how it looks in practice
with a real file API.
The full source code can be seen here:\ifthenelse{\boolean{anon}}{
    \textbf{link redacted for review, open fs.html}
}{
    \href{\formalisationURL/fs.html}{\formalisationURL/fs.html}.
}

First, we need functions to perform the IO operations on the filesystem.
I am going to rely on Idris' built-in functionality:

\noindent
\begin{minipage}[t]{0.45\linewidth}
\vspace{-1em}
\begin{align*}
    &\texttt{fPutStrLn} : \texttt{File} \to \texttt{String} \to \IO\ \top\\
    &\texttt{open} : \texttt{String} \to \IO (\texttt{FileError} + \top)\\
    &\texttt{close} : \texttt{File} \to \IO (\texttt{FileError} + \top)
\end{align*}
\end{minipage}
\hfill
\vline
\hfill
\begin{minipage}[t]{0.45\linewidth}
\vspace{-1em}
\begin{align*}
    &\texttt{data}\ \mathit{OpenFile} = \mathit{MkOpen}\ \String\\
    &\texttt{data}\ \mathit{WriteFile} = \mathit{MkWrite}\ \String\ \File\\
    &\texttt{data}\ \mathit{CloseFile} = \mathit{MkClose}\ \File
\end{align*}
\end{minipage}\\[0.3em]

With those, I implement the costate using the $costate$ function~\ref{def:costate} that acts as the
interpreter of the previously defined API.
\begin{align*}
    &\FS : Container \to Container\\
    &\FS = Lift (\IO \circ (\texttt{FileError} + \_))\\
    &writeMany : Costate (\FS\ WriteC*) = costate\ \texttt{fPutStrLn}\\
    &openFile : Costate (\FS\ OpenC) = costate\ \texttt{open}\\
    &closeFile : Costate (\FS\ CloseC ) = costate\ \texttt{close}\\
    &combined : Costate (\FS\ (OpenC \eseq WriteC* \eseq CloseC))\\
    &combined = \mathit{seqM}\ (\mathit{seqM}\ openFile\ writeMany)\ closeFile
\end{align*}

To run the program, we need to supply a valid input defined by
the type of the session.
\begin{align*}
    & writeOnce, writeTwice, writeNone : \mathit{StarShp}\ \mathit{WriteC}\\
    & writeNone = Done\\
    & writeOnce = single (\mathit{MkWrite}\ file\ "hello")\\
    & writeTwice = More (\mathit{MkWrite}\ file\ "hello") (\lambda \_. single\ (\mathit{MkWrite}\ file\ "world"))
\end{align*}

We can embed any sequence of writes into a valid session by surrounding it with
an opening and closing request. We turn it into a $State$ using definition~\ref{def:state}.
\begin{align*}
    &\mathit{mkSession} : \mathit{StarShp}\ WriteC \to \mathit{State\ (OpenC \eseq WriteC* \eseq CloseC)}\\
    &\mathit{mkSession}\ writes = state (\mathit{MkOpen}\ "append", \lambda file . (writes ,\lambda\_ .
                   \mathit{MkClose} file)
\end{align*}

The pair of a state and costate forms a \emph{context} that we can run.
We run it using $run$~\ref{eq:runner}.
\begin{align*}
    & \mathit{main} : \IO\ \top\\
    & \mathit{main} = \mathit{run}\ (\mathit{mkSession}\ \mathit{writeOnce})\ \mathit{combined}
\end{align*}

And as expected, we can run all sorts of sessions that are also valid, like $writeNone$ and $writeTwice$
written above, using the same runner.

\section{Conclusion \& Future work}

As I have shown in the two demos, using containers and their morphisms not only
forms a compelling mathematical basis to study large compositional systems
but also provides the programming infrastructure to build such systems.

From here, I see three different paths of further improvements.
One could be to enhance the practical usability of the framework, build binders
to more libraries, database systems, user interfaces, operating systems etc. Such
that it becomes possible to build a broad range of programs, all the way from the
TCP layer up to the GUI.

Another path could be to drill down on one specific use-case and elucidate problems
that the framework encounters and solutions it provides as a way to further both
the theory and the practice of the tool. Uses cases I am interested
in include compilers, machine learning, and video game development.

Finally, the last axis of development would be to expand the theory with components
from other systems. For example, developing a type theory
for bidirectional systems, incorporating session types, or discussing the para-construction
, all of which would contribute to making the research broader and compatible with other areas
of research, such as cybernetics or categorical deep learning.

\ifthenelse{\boolean{anon}}{}{
    \subsection{\ackname}
    \begin{credits}
    I would like to thank my colleagues from MSP who offered deep and practical insight.
    In particular Bob, Fred, Guillaume, Sean who gave me precious feedback, Jules and
    Zanzi for encouraging me to write those ideas down, Shin-ya for teaching me how
    to think in categories, and the team at NPL for supporting my work both emotionally
    and financially.
    \end{credits}
}
\bibliographystyle{splncs04}
\bibliography{./bibliography.bib}
\appendix
    
\section{MaybeAll is a functor}\label{sec:maybe-functor-proof}

The functor is given by $MaybeAll$ on objects and given a morphism
$(f \lhd f')$ the corresponding lifted morphism is given by
$mapMor (f\lhd f') = (map\ f \lhd mapAll\ f')$

\lemma{$mapMor$ respects identity}

We want to prove that $mapMor (id \lhd id) = (map\ id \lhd mapAll\ id) = (id \lhd id)$.

Because $Maybe$ is a functor we have that $map\ id \equiv id$. For
$mapAll\ id\ v\ x \equiv x$ we need to proceed by case-analysis.

\begin{itemize}
\item When $v = Nothing$ in $mapAll\ id\ Nothing\ () = ()$
\item When $v = Just\ v$ in $mapAll\ id\ (Just\ v)\ w = id\ w = w$
\end{itemize}

\lemma{$mapMor$ respects composition}

Given two morphisms $(f \lhd f') : (a \rhd a') \Rightarrow (b \rhd b')$ and
$(g \lhd g) : (b \rhd b') \Rightarrow (c \rhd c')$ we want to prove that
$mapMor (g \lhd g') \circ (mapMor (f \lhd f')) = mapMor ((g \lhd g') \circ (f \lhd f'))$

By evaluating $(g \lhd g') \circ (f \lhd f') = (g \circ f \lhd \lambda x. f'\ x \circ (g'\ (f\ x)))$
we obtain the term $(map (g \circ f) \lhd mapAll (\lambda x. f'\ x \circ (g'\ (f\ x))))$.\\
By expanding
the definition of $mapMor$ we obtain on the left side
$(map\ g \lhd mapAll\ g') \circ (map\ f \lhd mapAll\ f')$ by definition of container morphism composition we obtain
$(map\ g \circ map\ f \lhd \lambda x. mapAll\ f'\ x \circ mapAll\ g'\ (map\ f\ x))$

Because $Maybe$ preserves identities, we have $map\ g \circ map\ f \equiv map\ (g \circ f)$. It is left to prove
that $\forall x, y . mapAll\ f'\ x (mapAll\ g'\ (map\ f\ x)) y) = (\lambda z. f'\ z \circ (g'\ (f\ z)))\ x\ y$.

We do this by case-analysis on $x$:

\begin{itemize}
\item When $x = Nothing$ then \\$mapAll\ f'\ Nothing\ (mapAll\ g'\ Nothing\ ())$\\$= mapAll\ f'\ Nothing\ ()$\\$ = () = mapAll\ (\lambda z. f'\ z \circ (g'\ (f\ z)))\ Nothing\ () $\\
\item When $x = Just\ v$ then \\
$mapAll\ f'\ (Just\ v) (mapAll\ g'\ (Just\ (f\ v))\ w) $\\$= mapAll\ f'\ (Just\ v) (g'\ (f v)\ w)$\\$= f'\ v\ (g'\ (f\ v)\ w)$\\$ = mapAll\ (\lambda z. f'\ z \circ (g'\ (f\ z))) (Just\ v)\ w$
\end{itemize}

\section{Lemmas for Kleene as Functor}\label{apx:functor}

\begin{lemma}[$mapShp$ preserves identities]\\
$\forall (a : Container). \forall (x : StarShp\ a). mapShp_{id} x = x$.
\end{lemma}

\begin{proof}
We perform a proof by induction on $x$.\\
\begin{itemize}
\item In the base case where $x = Done$ we have $map_{id} Done = Done$, which completes the proof.
\item In the inductive case where $x = More\ x_1\ x_2$ then $map_{id} (More\ x_1\ x_2)$ evaluates to
$More\ x_1 (\lambda v. mapShp_{id} (x_2\ v))$. The first argument is the same, the second is
equal by induction $mapShp_{id} \circ x_2 = x_2$.
\end{itemize}
\end{proof}

\begin{lemma}[$mapPos$ preserves identities]\\
$\forall (a: Container). \forall (x: StarShp\ a) . \forall (y : StarPos\ a\ x). mapPos_{id}\ x\ y = y$
\end{lemma}
\begin{proof}We perform a proof by induction on $x$\\
\begin{itemize}
    \item In the base case where $x = Done$ then $y = \top$ and $mapPos_{id} Done \top$ evaluates to $\top$ completing the proof.
    \item In the inductive case where $x = More\ x_1\ x_2$ then the goal becomes\\
    $mapPos_{id} (More\ x_1\ x_2) (y_1, y_2) = (y_1, y_2)$, the left side evaluates to\\
    $(y_1, mapPos_{id}\ y_1\ y_2)$, which is
    identical to the goal by induction on $(y_1, y_2)$.
\end{itemize}
\end{proof}

\begin{lemma}[$mapShp$ preserves composition]\\
$\forall(a, b, c : Container) (f : a\Rightarrow b) (g : b \Rightarrow c). \forall(x : StarShp\ a). mapShp_{g\circ f} x = mapShp_g (mapShp_f\ x)$
\end{lemma}
\begin{proof}
    \setlength{\leftmargini}{0pt} 
\setlength{\itemsep}{5pt} 
\setlength{\parsep}{5pt} 
    We define $f = f_1 \lhd f_2, g = g_1 \lhd g_2$ to refer to the forward and backward
maps of $f$ and $g$, then proceed by induction on $x$ as above.\\
\begin{itemize}
\item In the base case where $x = Done$ then $mapShp_{g\circ f} Done$ evaluates to $Done$ completing the proof.
\item In the inductive case where $x = More\ x_1\ x_2$ then we have:

    \[
    \renewcommand{\arraystretch}{1.5} 
\begin{array}{rl}
& mapShp_{g\circ f} (More\ x_1\ x_2)\\
    & \quad \text{(By definition of $f ; g$)}\\
\equiv & mapShp_{g \circ f \lhd f' \circ g'}\ (More\ x_1\ x_2) \\
    & \quad \text{(By definition of $mapShp$)}\\
\equiv & More (g (f(x_1))) (mapShp_{g\circ f} \circ x_2 \circ f'_{x_1} \circ g'_{f(x_1)})\\
    & \quad \text{(By induction)}\\
\equiv & More (g (f(x_1))) (mapShp_g \circ mapShp_f \circ x2 \circ f'_{x_1} \circ g'_{f(x_2)}) \\
     & \quad\text{(By definition of $mapShp_g$)}\\
\equiv & mapShp_g (More (f(x_1)) (mapShp_f \circ x_2 \circ f'_{x_1}) \\
    & \quad \text{(By definition of $mapShp_f$)}\\
\equiv & mapShp_g (mapShp_f (More\ x_1\ x_2))
\end{array}
\]
\end{itemize}
    \renewcommand{\arraystretch}{1} 
\end{proof}

\begin{lemma}[$mapPos$ preserves composition]\\
$\forall(a, b, c : Container) (f : a \Rightarrow b) (g : b \Rightarrow c). \\\forall (x : StarShp\ a).
\forall (y : StarPos_{g\circ f} (mapShp_{g\circ f}\ x)). mapPos_{g\circ f}\ y = $\\$mapPos_g\ (mapPos_f\ y)$
\end{lemma}
\begin{proof}
    Like above, we perform the proof by induction on $x$.
    \begin{itemize}
        \item In the base case where $x = Done$, then $mapPos_{g \circ f}\ Done\ \top$ evaluates to $\top$ and so does $mapPos_f\ Done\ (mapPos_g\ Done\ \top)$ completing the proof.
        \item In the inductive case where $x = More\ x1\ x2$, then the goal becomes \\
            $mapPos_{g \circ f}\ (More\ x1\ x2)\ (y1,y2) = $\\
            $mapPos_f (More\ x1\ x2)\ (mapPos_g (mapShp_f (More\ x1\ x2)) (y1, y2))$
    \end{itemize}
    \renewcommand{\arraystretch}{1.5} 
    \[
\begin{array}{rl}
         & mapPos_{g\lhd g' \circ f\lhd f'} (More\ x1\ x2) (y1 , y2)\\
    \equiv & \quad\text{(By definition of $g\lhd g' \circ f\lhd f'$)}\\
         & mapPos_{g \circ f\lhd \lambda x. f' x \circ (g' (f x))} (More\ x1\ x2) (y1 , y2)\\
    \equiv & \quad\text{(By definition of $mapPos_{g \circ f\lhd \lambda x. f' x \circ (g' (f x))}$)}\\
         & (f'\ y1\ (g'\ (f\ y1)), mapPos_{g\circ f}\ (x2\ (f'\ y1 (g'\ (f\ y1))))\ y2) \\
    \equiv & \quad\text{(By induction on the second projection)} \\[-8pt]
         & \parbox{0.8\textwidth}{\begin{align*}
        (f'\ y1\ (g'\ (f\ y1)), & mapPos_f\ (x2\ (f'\ x1\ (g'\ (f\ x1)\ y1))) \\
                                & (mapPos_g\ (mapStarShp_f\ (x2\ (f'\ x1\ (g'\ (f\ x1)\ y1))))\ y2))
         \end{align*}}\vspace{-7pt}\\
    \equiv & \quad \text{(By definition of $mapPos_f$)} \\
         & \parbox{\textwidth}{\begin{align*}mapPos_f\ & (More\ x1\ x2) (g'\ (f\ x1)\ y1 ,\\
                                                       & mapPos_g\ (mapStarShp_f (x2\ (f\ x1\ (g'\ (f\ x1)\ y1)))) y2)\end{align*}}\\
    \equiv & \quad \text{(By definition of $mapPos_g$)}\\
         & \parbox{\textwidth}{\begin{align*}mapPos_f\ & (More\ x1\ x2) \\
                                                       & (mapPos_g\ (More\ (f\ x1)\ (mapStarShp_f\circ x2 \circ (f'\ x1))))) (y1 , y2))\end{align*}}\\
    \equiv & \quad \text{(By definition of $mapShp_f$)}\\
         & mapPos_f\ (More\ x1\ x2) (mapPos_g\ (mapShp_f (More\ x1\ x2)) (y1, y2))
\end{array}
\]
\end{proof}

\end{document}